\documentclass[letterpaper, 10 pt, conference]{ieeeconf}  

\IEEEoverridecommandlockouts                              

\usepackage{amsfonts, dsfont, mathtools, amsmath, amssymb, bbold, mathrsfs, bm, upgreek} 
\usepackage{graphicx, float, epsfig, color, psfrag, adjustbox, enumerate, svg}
\usepackage[font=small]{caption}
\usepackage{refcount, url, cleveref}
\usepackage[noadjust]{cite}
\usepackage[usenames,dvipsnames,svgnames,table]{xcolor}
\usepackage{soul}
\usepackage[normalem]{ulem} 

\title{\LARGE \bf
On Feedback Speed Control for a Planar Tracking 
}

\author{Xincheng Li, Tengyue Liu, Udit Halder
\thanks{All authors are with the Department of Mechanical and Aerospace Engineering, University of South Florida. 
  Corresponding e-mail:  {\tt\small udithalder@usf.edu}}
}
\def\R{{\mathds{R}}}

\def\0{{\mathbb{0}}}
\def\1{{\mathds{1}}}

\newcommand{\norm}[1]{\left\lVert#1\right\rVert}

\newcommand{\abs}[1]{\left| #1 \right|}

\definecolor{db}{RGB}{23,20,119}
\definecolor{dg}{RGB}{2,101,15}

\newtheorem{proposition}{Proposition}[section]

\newtheorem{remark}{Remark}
\newtheorem{assumption}{Assumption}

\usepackage{upgreek}

\newcommand{\transpose}{\intercal}

\newcommand{\material}[1]{
	\ifthenelse{\equal{#1}{\kappa}}{\upkappa}{
	\ifthenelse{\equal{#1}{\nu}}{\upnu}{
	\ifthenelse{\equal{#1}{\omega}}{\upomega}{
	\ifthenelse{\equal{#1}{\sigma}}{\upsigma}{
	\ifthenelse{\equal{#1}{\theta}}{\uptheta}{
	\mathsf{#1}}}}}}
}

\newcommand{\utangent}{\mathbf{x}}
\newcommand{\unormal}{\mathbf{y}}

\graphicspath{{figures/}}

\begin{document}
\bstctlcite{BSTcontrol} 
\maketitle
\thispagestyle{empty}
\pagestyle{empty}


\begin{abstract}
This paper investigates a planar tracking problem between a leader and follower agent. We propose a novel feedback speed control law, paired with a constant bearing steering strategy, to maintain an abreast formation between the two agents. 
We prove that the proposed control yields asymptotic stability of the closed-loop system when the steering of the leader is known. For the case when the leader's steering is unavailable to the follower, we show that the system is still input-to-state stable with respect to the leader's steering viewed as an input. Furthermore, we demonstrate that if the leader's steering is periodic, the follower will asymptotically converge to a periodic orbit with the same period. We validate these results through numerical simulations and experimental implementations on mobile robots. Finally, we demonstrate the scalability of the proposed approach by extending the two-agent control law to an $N$-agent chain network, illustrating its implications for directional information propagation in biological and engineered flocks.
\end{abstract}


\section{Introduction} \label{sec:intro}

In the ubiquitous displays of collective motion in nature~-- from the flocking of birds to the schooling of fish to the swarming of insects and krills, 
coherent spatial patterns and directional information propagation are often observed~\cite{couzin2002collective, katz2011inferring, attanasi2014information, Murphy2019}.  
Inspired by such phenomena, a large body of work in control theory and robotics has investigated distributed algorithms for flocking and formation control in multi-agent systems~\cite{olfati2006flocking,sepulchre2007stabilization, paley2007oscillator}. 

A widely studied architecture in engineered multi-agent systems is the `leader–follower paradigm', where a leader agent follows a prescribed trajectory while follower agents regulate their motion using feedback in order to maintain desired relative positions or orientations \cite{tang2021formation, trinh2021finite}. Furthermore, two-agent pursuit in falcons, raptors, or dragonflies
~\cite{brighton2017, mischiati2015internal} has inspired a body of work~\cite{justh2006steering,halder2016steering,Zhang2004}, where the dynamics of two-agents are studied from a pure geometric framework, using nonholonomic models of the individuals (e.g., the unicycle model~\cite{lavalle2006planning}) and relative shape to construct feedback control. 
However, in most of these works, the emphasis is placed on \textit{steering control}~\cite{paley2007oscillator, OKeeffe2017}.

In contrast, the role of \textit{speed regulation} in formation maintenance has received comparatively less attention {\cite{consolini2008leader, do2007formation}}. In many biological systems, individuals adjust both their direction and speed in response to the motion of neighbors~\cite{katz2011inferring, Murphy2019}. In addition, empirical studies have reported a consistent relationship between speed and path curvature (closely related to the steering control), often referred to as the {speed--curvature trade-off}, across a range of biological movements~\cite{flash1985coordination, zago2018speed}. These observations suggest that speed modulation plays an important role in maintaining spatial organization in biological collectives.

Motivated by these observations, this paper studies a planar leader–follower tracking problem, with particular focus on the feedback speed control of the follower agent. Each agent is modeled as a planar unicycle, whose speed and steering inputs are considered controls. The follower seeks to track the leader and maintain an ``abreast'' formation. 
One of the novelties of this work, in contrast to the existing approaches to this problem~\cite{consolini2008leader, do2007formation}, lies in the study of \textit{leader independence}, where the follower does not have perfect knowledge of the leader's instantaneous control inputs. 
%

\smallskip
The specific contributions of this paper are as follows.

(1) 
A major contribution of this work is the proposal of a feedback speed control law for the follower 
accompanied by a \textit{constant bearing} steering control law~\cite{galloway2013symmetry}. 
The two controls are {decoupled} using a singular perturbation argument, allowing for more emphasis to be placed on the speed control.
We show that the proposed controls (assuming full knowledge of leader's steering) asymptotically stabilize the closed-loop system at the desired abreast formation (Prop.~\ref{prop:1}). 

(2) More realistically, when the leader's steering is not known to the follower, 
we prove that the closed-loop system is input-to-state stable~\cite{khalil2002nonlinear} with respect to the leader's steering input (Prop.~\ref{prop:2}). Furthermore, we also analyze a special case where the leader's steering is periodic, and show that the closed-loop system asymptotically converges to a periodic orbit with the same period (Prop.~\ref{prop:3}). 

(3) We validate the theoretical results using both numerical simulations and robotic implementations. 

(4) Finally, we provide a plausible passage to a multi-agent formation control and illustrate its relation to information transfer in a flock.   

\smallskip
The remainder of this paper is organized as follows. The modeling and problem formulation are given in Sec.~\ref{sec:problem}, control design and its analysis are provided in Sec.~\ref{sec:control}, followed by the numerical and experimental results in Sec.~\ref{sec:results}. An extension to an $N$ agent chain network is discussed in Sec.~\ref{sec:nagents} and the paper is concluded in Sec.~\ref{sec:conclusion}.
\section{Problem Formulation} \label{sec:problem}
\subsection{Agent Model}
We consider two agents moving in the plane ($\R^2$). Let $\mathbf{r}_i \in \R^2$ denote the position vector of agent $i$ ($i=1,2$). To describe the orientation, we attach a moving frame to each agent, denoted by a unit tangent vector $\utangent_i$ (defining the heading angle of the agent) and a unit normal vector $\unormal_i$ (see Fig.~\ref{fig:formation}). The dynamical equations for the trajectories of the agents are then given by \cite{Bishop1975}:
\begin{align} \label{eq:agent_kinematics}
        \dot{\mathbf{r}}_i = v_i \utangent_i, ~~
        \dot{\utangent}_i = u_i \unormal_i, ~~
        \dot{\unormal}_i = -u_i \utangent_i 
\end{align}
where $v_i$ is the linear speed and $u_i$ is the angular speed. We consider the variables $v_i$ and $u_i$ as speed and steering control inputs, respectively. We designate agent 1 as the \textit{leader} with some given open-loop program $v_1, u_1$, and agent 2 as the \textit{follower} (with corresponding controls $v_2, u_2$). Here, we make the following mild assumptions on the bounds of the leader's controls. 
\begin{assumption} \label{assumption:bounds}
Assume that $ v_\epsilon \leq v_1 (t) \leq k_v$ and $\abs{u_1(t)} \leq k_u$, for some bounds $k_v, k_u, v_\epsilon > 0$ and for all $t\geq 0$.
\end{assumption}

\begin{figure}[t]
    \centering
    \includegraphics[width=\linewidth]{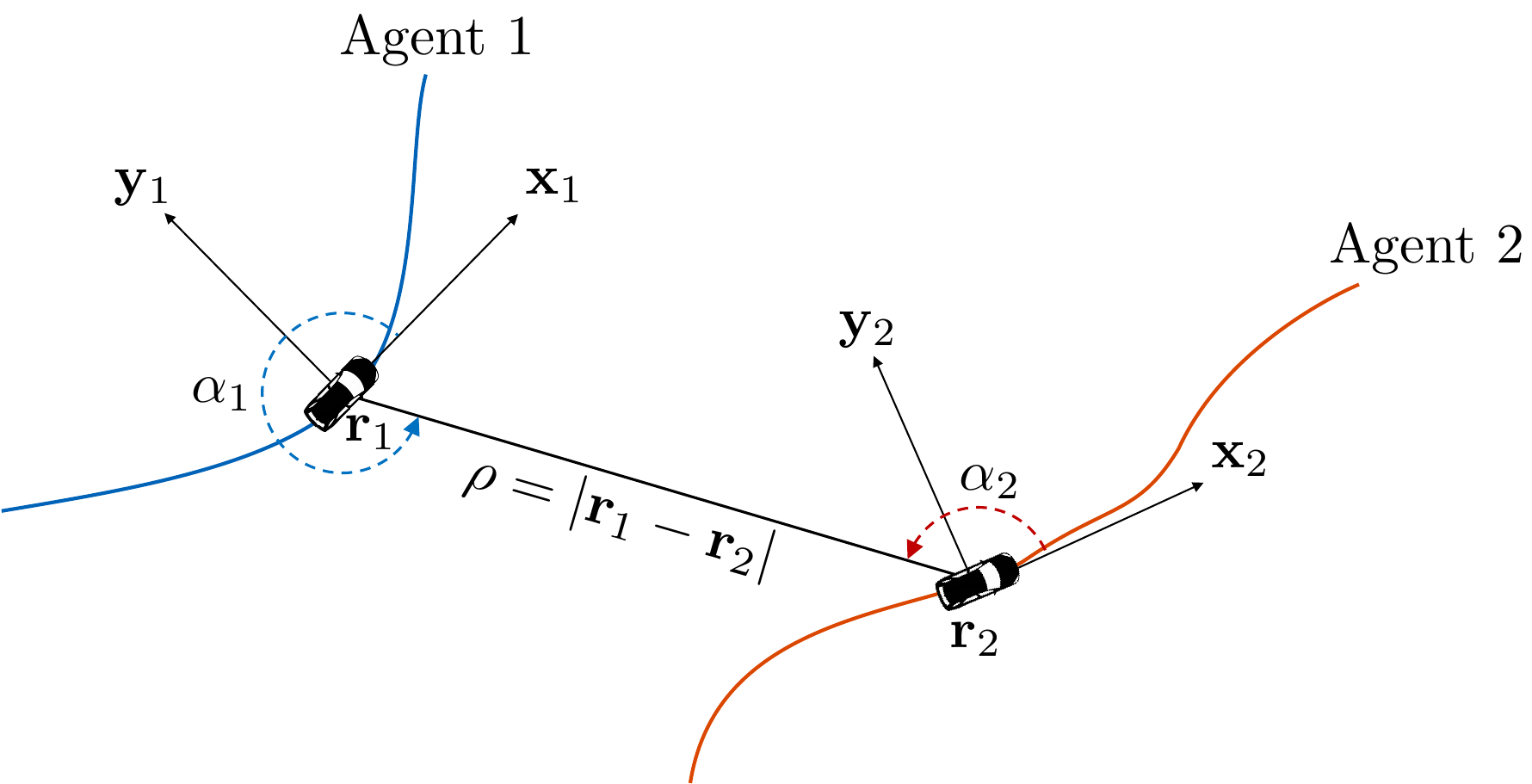}
    \caption{Modeling the motion and formation of two agents using their relative position and angles.}
    \label{fig:formation}
\end{figure}

\subsection{Shape Dynamics}

To analyze the formation independently of the global frame, three shape variables are introduced: the relative distance $\rho$ and the relative bearing angles $\alpha_1, \alpha_2$ (see Fig.~\ref{fig:formation}). These variables are defined as follows. 

Let $\mathbf{r}_{ij} := \mathbf{r}_i - \mathbf{r}_j, ~ i,j \in \{1,2\}$ be the relative position vector of agent $i$ with respect to agent $j$. Then, the shape variables $\rho, \alpha_1, \alpha_2$ are defined as
\begin{align}
    \begin{split}
        \rho := \norm{\mathbf{r}_{12}} &= \norm{\mathbf{r}_{21}} \\
        R(\alpha_1)\mathbf{x}_1 := \frac{\mathbf{r}_{21}}{\rho}, \qquad& 
        R(\alpha_2)\mathbf{x}_2 := \frac{\mathbf{r}_{12}}{\rho} 
    \end{split}
\end{align}
where $R(\theta) \in SO(2)$ is a rotation matrix that rotates any planar vector counterclockwise by an angle $\theta$.

Differentiating these definitions yields the shape dynamics \cite{halder2016steering, galloway2013symmetry}
\begin{align} \label{eq:shape_dynamics}
    \begin{split}
        \dot{\rho} &= -v_1 \cos \alpha_1 - v_2 \cos \alpha_2 \\
        \dot{\alpha}_1 &= - u_1 + \frac{1}{\rho} (v_1 \sin \alpha_1 + v_2 \sin \alpha_2)  \\
        \dot{\alpha}_2 &= - u_2 + \frac{1}{\rho} (v_1 \sin \alpha_1 + v_2 \sin \alpha_2) 
    \end{split}
\end{align}
These equations describe the evolution of the formation geometry entirely in terms of relative quantities and control inputs.

\subsection{Control Problem} \label{sec:control_problem}
The control objective is to stabilize the follower to a desired formation defined by a fixed distance $\rho_0$ and two given relative angles. In particular, we seek control laws $(v_2, u_2)$ that asymptotically stabilize the shape variables to an ``abreast" formation, where $\alpha_1 = \pm\frac{\pi}{2}$ and $\alpha_2 = \mp \frac{\pi}{2}$. 
In other words, we look to construct a formation where the two agents become Bertrand mates \cite{Bertrand1850, Zhang2004}.

The signs of $\alpha_1$ and $\alpha_2$ depend on the relative positions of the agents. For example, if the follower is to the right of the leader (as in Fig.~\ref{fig:formation}), $\alpha_1 = - \frac{\pi}{2}, \alpha_2 = + \frac{\pi}{2}$, in the abreast formation. For the purposes of this paper, we will follow this convention. Indeed, the signs need to be flipped for the other case, and can be dealt with analogously. 

Formally, we wish to stabilize the shape dynamics \eqref{eq:shape_dynamics} around the equilibrium point $z^* = \left(\rho_0, -\frac{\pi}{2}, \frac{\pi}{2} \right)$.

\vspace*{-0pt}
\section{Control Design}\label{sec:control}

At the outset, let the steering control $u_2$ be chosen as
\begin{align}
    u_2 &= -\mu_2 \cos\alpha_2 + \frac{1}{\rho}(v_1 \sin \alpha_1 + v_2 \sin \alpha_2)
    \label{eq:u2_CB}
\end{align}
where $\mu_2 > 0$ is a control gain. This steering control law is a standard \textit{constant bearing} (CB) strategy \cite{galloway2013symmetry}, a concept rooted in classical guidance and navigation. Unlike \textit{classical pursuit} \cite{Davis1962}, where the follower (pursuer) points directly at the leader (target) ($\alpha_2 \rightarrow 0$), CB control maintains a fixed non-zero bearing angle (here $ \alpha_2 \rightarrow \frac{\pi}{2}$ as argued below).

Under the steering control~\eqref{eq:u2_CB}, the $\alpha_2$ equation reads as $\dot{\alpha}_2 = \mu_2 \cos \alpha_2$. It is a straightforward exercise to show that $\sin \alpha_2 (t) = \tanh (\mu_2 t + \tanh^{-1}(\sin \alpha_2 (0)))$. Therefore, for all initial conditions $\alpha_2 (0) \neq - \frac{\pi}{2}$,~$\sin \alpha_2 (t) \rightarrow 1$  (i.e., $\alpha_2 (t) \rightarrow \tfrac{\pi}{2}$) as $t \rightarrow \infty$. Moreover, owing to the nature of the $\tanh (\cdot)$ function, this convergence is exponentially fast, with higher gain $\mu_2$ resulting in faster convergence.

Therefore, it is natural to restrict
the shape dynamics~\eqref{eq:shape_dynamics} to the $\alpha_2 = \tfrac{\pi}{2}$ manifold, which yields
\begin{align} \label{eq:shape_dynamics_reduced}
    \begin{split}
        \dot{\rho} &= -v_1 \cos \alpha_1\\
        \dot{\alpha}_1 &= - u_1 + \frac{1}{\rho} (v_1 \sin \alpha_1 + v_2) 
    \end{split}
\end{align}
\begin{remark} {\bf Singular perturbation and time-scale separation}. The restriction of the fully coupled system to the $\alpha_2 = \tfrac{\pi}{2}$ manifold is formally justified via singular perturbation theory \cite[Chapter 11]{khalil2002nonlinear}. This decoupling requires the assumption that the follower's steering gain $\mu_2$ is chosen to be sufficiently large (compared to the speed gain, yet to be defined). 
Under this setting, we can define a small perturbation parameter $\epsilon = \tfrac{1}{\mu_2} > 0$, creating a strict two-time-scale separation. The $\alpha_2(t)$ dynamics act as the fast boundary layer, converging exponentially to the manifold $\alpha_2 = \tfrac{\pi}{2}$, which allows the remaining shape variables $(\rho, \alpha_1)$ to be rigorously analyzed as the reduced slow system.
\end{remark}

It remains to specify the feedback control $v_2$ to stabilize the distance to a given value $\rho_0$ and the angle $\alpha_1$ to $-\tfrac{\pi}{2}$. For this purpose, we introduce a function $g : (0, \infty) \rightarrow \R$ and its derivative $f: (0, \infty) \rightarrow \R$ as follows:
\begin{itemize}
    \item $g(\cdot)$ is $C^2$ and positive definite, i.e., $g(\rho) > 0, \forall \rho \neq \rho_0$, $g(\rho_0) = 0$,
    \item $\lim_{\rho \to 0} g(\rho) = \infty$ and $\lim_{\rho \to \infty} g(\rho) = \infty$,
    \item $\frac{d}{d\rho}g(\rho) = f(\rho)$, where $f(\cdot)$ is $C^1$. 
\end{itemize}
We are now ready to state the following result.
 \begin{proposition} \label{prop:1}
Consider the reduced shape dynamics (\ref{eq:shape_dynamics_reduced}). Then, the  feedback speed control given by
\begin{align}
    v_2 &= -v_1 \sin \alpha_1 + \rho \left( u_1 - \mu_1 \cos\alpha_1 + v_1 f(\rho) \right) \label{eq:v2_ideal}
\end{align}
with gain $\mu_1> 0$, renders the equilibrium $(\rho_0, -\tfrac{\pi}{2})$ (locally) asymptotically stable.
\end{proposition}

\begin{proof}
Restricting the dynamics to the $\alpha_2 = \frac{\pi}{2}$ manifold, the reduced subsystem $z_1 = (\rho, \alpha_1)$ under the speed control (\ref{eq:v2_ideal}) becomes
    $\dot{z}_1 =
    \begin{bmatrix}
        -v_1\cos\alpha_1 \\
        -\mu_1\cos\alpha_1 + v_1 f(\rho)
    \end{bmatrix}$.
Consider the positive definite Lyapunov candidate $V_1 = 1+\sin\alpha_1 + g(\rho)$, which is zero only at the equilibrium $z_1^* = (\rho_0, - \frac{\pi}{2})$. 
Its time derivative $\dot{V}_1 = -\mu_1 \cos^2\alpha_1 \leq 0$ is negative semi-definite, with the function equal to zero on the set $\{(\rho, \alpha_1) \, | \, \alpha_1 = -\frac{\pi}{2}\}$, assuming we are in a neighborhood of the equilibrium point such that $V_1(z_1) < 2$. On this set, $\dot{\alpha}_1 = 0$ implies $f(\rho) = 0$, i.e., $\rho = \rho_0$. So, the equilibrium point $z_1^*= (\rho_0, -\frac{\pi}{2})$ is the largest invariant set. By LaSalle's invariance principle~\cite[Theorem 4.4]{khalil2002nonlinear}, the system is locally asymptotically stable around the equilibrium point $z_1^*$, with the region of attraction the neighborhood such that~$V_1(z_1) < 2$.
\end{proof}

\smallskip
\begin{remark} {\bf Importance of speed feedback.} 
Speed regulation is important and, in fact, necessary to achieve the desired tracking. In \cite{Zhang2004}, for instance, a similar two-agent approach is used to achieve boundary following. In that work, however, the follower maintains a constant speed, and the leader is a `virtual shadow point' which moves along the boundary at variable speeds determined by the boundary's geometry and the follower's motion (see~\cite[Eq.~(15)]{Zhang2004}). 
This paper in contrast proposes an explicit follower speed control which regulates both the bearing $\alpha_1$ and the distance $\rho$. 
\end{remark}

\smallskip
\subsection{Leader Independence}

In biological flocking or decentralized swarms, agents typically only have reliable information on their own speed and turning rate, alongside relative geometry (such as distance) of neighbors \cite{katz2011inferring, Murphy2019}. However, our feedback control laws~\eqref{eq:u2_CB}, \eqref{eq:v2_ideal} rely on the follower agent possessing perfect knowledge of all the shape variables -- $\rho, \alpha_1, \alpha_2$, the leader's instantaneous steering $u_1$ and linear speed $v_1$. 

In practical scenarios, we may only assume that the follower knows its own speed and steering ($v_2$ and $u_2$) and can obtain estimates of the relative shape variables $\rho$ and $\alpha_2$ as well as their derivatives $\dot\rho$ and $\dot\alpha_2$ through onboard sensors such as cameras or LiDAR. Using these measurements, we can calculate 
\begin{align*}
    X = - \dot\rho - v_2 \cos\alpha_2, \quad   Y = \rho(\dot\alpha_2 + u_2) - v_2\sin\alpha_2
\end{align*}
The leader's linear speed and relative bearings ($v_1$, $\alpha_1$) can then be estimated as 
\begin{equation}
    v_1 = \sqrt{X^2 + Y^2}, \;\;\; \alpha_1 = \arctan \left({Y}/{X} \right) 
    \label{eq:estimates}
\end{equation}
This estimation problem, when also accounting for sensor measurement errors and noise, leads to a nonlinear filtering problem. Assuming that these errors are bounded, we detail how these estimation errors affect the stability of the feedback control in Appendix. \ref{appdx:estimation_analysis}.

As it can be observed from \eqref{eq:estimates}, obtaining estimates of $v_1$ and $\alpha_1$ involves taking first derivatives of sensory data. However, estimating $u_1$ will require second derivatives, a process that is prone to errors. This provides motivations for considering a leader steering-independent approach where the follower does not have access to $u_1$, while still attempting to maintain stability. 



The speed control law~\eqref{eq:v2_ideal} without the $u_1$ term reads
\begin{equation}
    v_2 = -v_1 \sin \alpha_1 + \rho \left( -\mu_1 \cos\alpha_1 + v_1 f(\rho) \right) \label{eq:v2_robust}
\end{equation}
This yields a closed-loop system which may be regarded as a non-autonomous control system with the leader steering $u_1(t)$ acting as a perturbation input. We can then state the following result on the stability of such a system.

\begin{figure*}[t]
	\centering
	\includegraphics[width=1\textwidth]{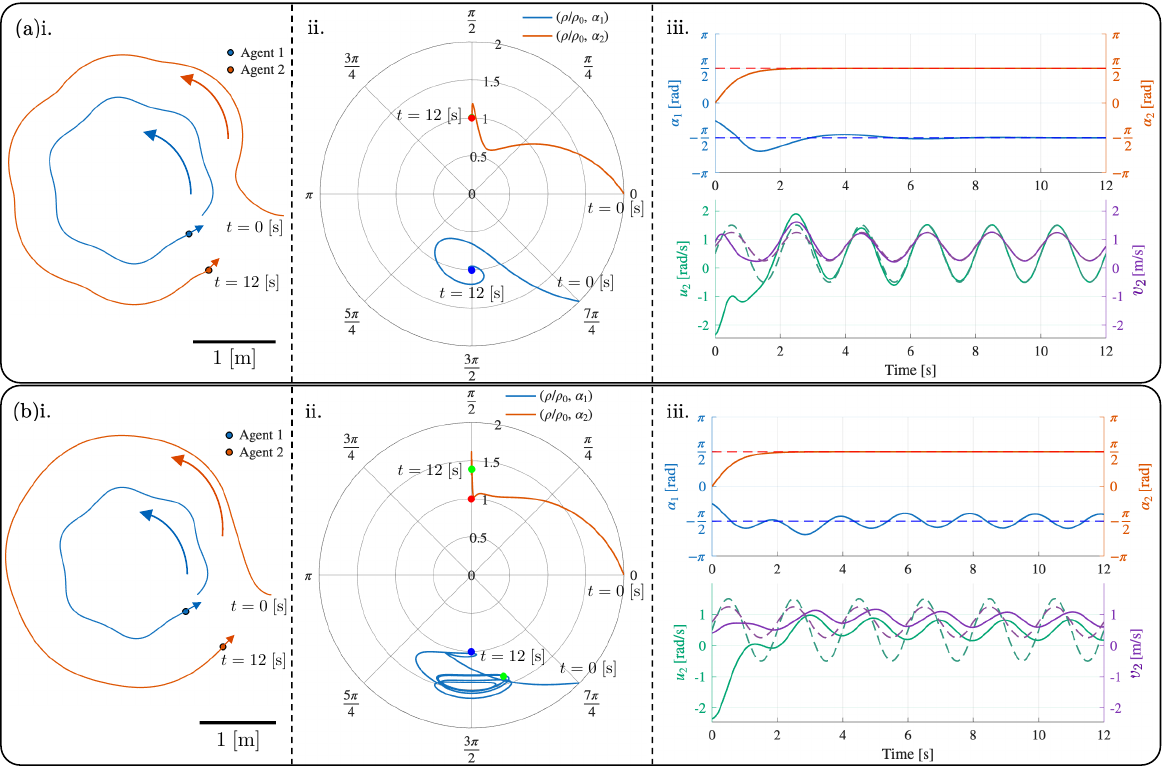}
	\caption{Formation control of two agents. Agent two (the follower in orange) is tasked with becoming a Bertrand mate moving to the right of agent one (the leader in blue). Agent two employs the controls laws as of Prop.~\ref{prop:1} in (a) and of Prop.~\ref{prop:2} in (b). The resulting trajectories are shown in the panels (i). Panels (ii) show the polar plots of the normalized $\rho$ and $\alpha_i, \;i= 1, 2$, and either convergence to the equilibrium $z^*$ (represented by the solid red and blue dots) in (a) or a periodic orbit in (b), with the green dot in (b)ii. showing the shape coordinate at $t = 12$ [s]. Panels (iii) show the convergence (a) or periodic oscillations (b) of $\alpha_1$ in the top subplot ($\alpha_2$ converges in both cases), and the controls of the follower agent in the bottom subplot. The dashed lines represent the corresponding desired values.}
	\label{fig:circular}
\end{figure*}

\begin{proposition} \label{prop:2}
The modified speed control law \eqref{eq:v2_robust} renders the closed-loop system \eqref{eq:shape_dynamics_reduced} input-to-state stable (ISS) with respect to the leader's steering $u_1$.
\end{proposition}


\begin{proof}\label{Proof: 3.2}
Under the modified speed control \eqref{eq:v2_robust}, the derivative of Lyapunov candidate ($V_1$) constructed in the proof for Prop. \ref{prop:1} 
is $\dot{V}_1 = -\mu_1 \cos^2\alpha_1 -u_1\cos\alpha_1$.
As in the proof for Prop. \ref{prop:1}, the zero-input system (i.e., $u_1 \equiv 0$) is asymptotically stable. Given $|u_1| \leq k_u$ (Assumption~\ref{assumption:bounds}), for $|\cos\alpha_1| \geq \frac{k_u}{\mu_1} \geq \frac{\abs{u_1}}{\mu_1}$, $\dot{V}_1$ will be negative semi-definite, with the largest invariant set the equilibrium point $z_1^* = (\rho_0, -\frac{\pi}{2})$. The 
ISS property follows then from \cite[Theorem 4.19]{khalil2002nonlinear}. 
\end{proof}

\smallskip
 Proposition~\ref{prop:2} shows that even without knowledge of the leader's steering, the follower will still maintain an approximate abreast formation with the leader, with some small error in the formation shape. However, the ISS property of the system ensures that the error will asymptotically decay to the equilibrium $(\rho_0, -\frac{\pi}{2})$ once the leader stops its turn and resumes a straight path ($u_1(t) \to 0$). 

Even if the leader steering $u_1(t)$ never goes to zero, we can still characterize the resulting relative geometry of the formation. In a special case, we show further that if the turning rate of the leader is a periodic function, the resulting shape will also converge to a periodic solution with the same period. 

\begin{proposition} \label{prop:3}
 Let the steering of the leader $u_1(t)$ be a continuous, $T$-periodic function, i.e., $u_1(t+T) = u_1(t)$ for all $t \ge 0$. Then, there exists a $k_u^* > 0$, 
 such that if $k_u < k_u^*$, then under the control law established in Prop.~\ref{prop:2}, the resulting shape dynamics of $(\rho(t), \alpha_1(t))$ possess a unique, locally exponentially stable $T$-periodic orbit in a neighborhood of the equilibrium $(\rho_0, -\frac{\pi}{2})$.
\end{proposition}

\begin{proof}\label{Proof: 3.1}
Under the speed control in Prop. \ref{prop:2} 
we can rewrite the $z_1$ subsystem as $\dot{z}_1 = h_1(z_1) + \epsilon h_2(t, z_1, \epsilon)$, i.e. a nominal unperturbed system given by
    $h_1(z_1) = \begin{bmatrix}
        -v_1\cos \alpha_1 \\
        -\mu_1\cos\alpha_1 + v_1f(\rho)
    \end{bmatrix}$, and the perturbation given by 
    $h_2(t, z_1, \epsilon) = \begin{bmatrix}
        0 \\
        -\frac{u_1 (t)}{\epsilon}
    \end{bmatrix}$, where $\epsilon = k_u$. 
Linearizing the unperturbed system ($u_1 \equiv 0$) around the equilibrium point $z_1^* = (\rho_0, -\frac{\pi}{2})$ gives the associated Jacobian matrix
$A= \big[\begin{smallmatrix}
 0 && -v_1 \\
        v_1f'(\rho_0) && -\mu_1
\end{smallmatrix}\big]$.
We see that $A$ is Hurwitz under our assumptions for $f(\rho)$; specifically since $\rho_0$ is a (local) minimum of the $g(\cdot)$ function, we have $g''(\rho_0) = f'(\rho_0) > 0$.
This implies that the equilibrium point $z_1^*$ is locally exponentially stable.

Then, by \cite[Theorem 10.3]{khalil2002nonlinear}, given a perturbation $u_1(t)$ that is continuous and $T$-periodic, 
there exists a $k_u^* >0$, such that $\forall\, k_u < k_u^*$, there exists a unique solution $\tilde{z}_1(t)$ that is also $T$-periodic. 
Moreover, $\tilde{z}_1(t)$ is exponentially stable.
\end{proof}

Proposition~\ref{prop:3} essentially limits the steering bound of the leader $k_u$ from being too large, upper bounded by $k_u^*$. The upper bound $k_u^*$ can be increased by increasing the gain $\mu_1$, 
thereby enlarging the robustness margin of the Jacobian $A$.
\begin{figure*}[t]
	\centering
    \includegraphics[width=1\textwidth]{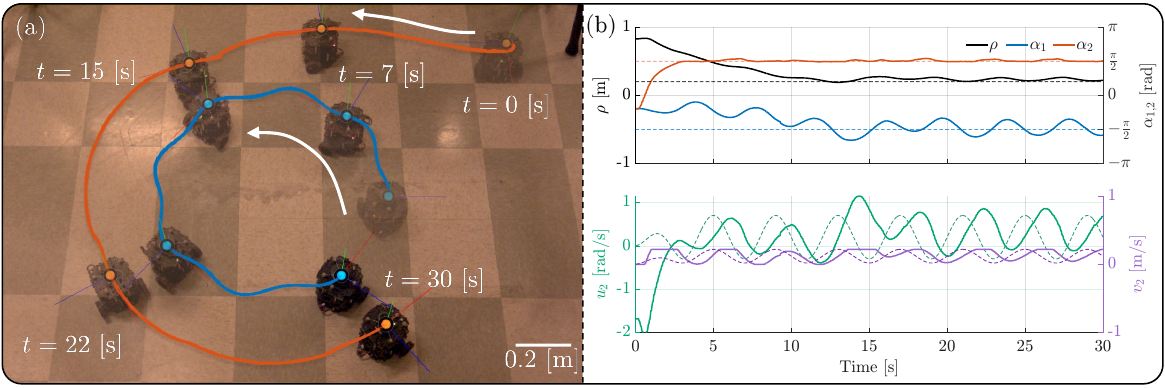}
    \caption{Formation control of two TurtleBot 3 Burger robots, implementing Prop. \ref{prop:2}. (a) Trajectories of the leader (blue) and follower (orange) robot from $t=0$ [s] to $t=30$ [s]. Increasing opacity of the robots indicates the progression of time, with white arrows denoting the general directions of motion. (b) Top subplot displays the evolution of $\rho$, $\alpha_1$, and $\alpha_2$. The equilibrium $z^*$ is represented by the dashed lines. Bottom subplot displays follower control inputs $u_2$ and $v_2$ (solid lines), compared to their desired values (dashed lines).}
	\label{fig:experiment}
\end{figure*}
\section{Results}\label{sec:results}
To validate the theoretical analysis, we simulate the system dynamics numerically and deploy robotic hardware. 

\subsection{Numerical Simulations} \label{sec:numerics}


The shape dynamics of the system~\eqref{eq:shape_dynamics}, along with feedback control laws described in Sec.~\ref{sec:control} are numerically integrated in MATLAB using a Runge-Kutta algorithm of 4-5th order (\textit{ode45}), with relative and absolute error tolerance set to $10^{-5}$ and $10^{-6}$, respectively. For the control laws \eqref{eq:u2_CB}, \eqref{eq:v2_ideal}, \eqref{eq:v2_robust}: the $f(\cdot)$ function is chosen as $f(\rho) = \mu_\rho \left(\frac{\rho^2 - \rho_0^2}{\rho^2}\right), \mu_\rho >0$,
and the gains are chosen as $\mu_1 = 1, \; \mu_2 = 2, \; \mu_\rho = 2$. Once the shape variables are obtained by integration, agent trajectories are reconstructed from relative geometry. 

 The leader initializes at the origin $\mathbf{r}_1(0) = [0, 0]^\transpose$ [m] with a initial heading of $\pi/4$ [rad], maintaining a constant linear speed of $v_1 (t) \equiv 0.5$ [m/s]. To test the robustness of the formation, the leader executes a time-varying steering $u_1(t) = 0.5 + \sin(\pi t)$ [rad/s]. The follower starts at $\mathbf{r}_2(0) = [0, 1]^\transpose$ [m] with an initial heading of $\pi$ [rad]. The objective for the follower is to achieve a Bertrand mate formation, positioning itself to the right of the leader with a desired separation distance of $\rho_0 = 0.5$ [m].

The simulation results are shown in Fig.~\ref{fig:circular}(a),
where the follower operates under the ideal control law (Prop.~\ref{prop:1}), utilizing full knowledge of the leader's steering. In this case, the shape dynamics converge smoothly to the equilibrium point $z^*$. The trajectory (i) and the evolution of the relative states and controls (ii, iii) demonstrate the asymptotic stability of the formation. In addition, the results show faster convergence of $\alpha_2$, compared to $\alpha_1$.

Conversely, Fig. \ref{fig:circular}(b) demonstrates the system's behavior when the follower employs the leader-independent control law (Prop.~\ref{prop:2}), lacking knowledge of the leader's steering. Despite this limited information, the follower still successfully tracks the leader. The formation error remains bounded, confirming the ISS property of the proposed controller under unmodeled leader maneuvers. The figure also shows the shape variables $\rho(t)$ and $\alpha_1(t)$ converging to a periodic orbit with a period of 2 [s], the same period as the leader's steering $u_1(t)$, demonstrating the results of Prop. \ref{prop:3}.

\subsection{Robotic Implementation}
To further demonstrate the control strategies, robotic experiments are performed, which consist of two TurtleBot 3 Burger mobile robots \cite{turtlebot3_manual}, with a maximum linear speed $v_{max}=0.22$ [m/s] and a maximum angular speed $u_{max}=2.84$ [rad/s]. Global pose information for both robots is provided by OptiTrack \cite{optitrack_website} using four motion capture cameras and a color video camera. The communication architecture is built on ROS 2, which receives the real-time pose information and sends control commands to the robots from a centralized Linux workstation. {We note here that local computation of the feedback control can be implemented. We however do not employ such a scheme here for simplicity.}

The experiments are conducted within a testbed of approximately $2.3 \times 2.0$ meters. The leader and the follower are initialized at $\mathbf{r}_1(0)=[0.287,0.005]^\transpose$ [m] and $\mathbf{r}_2(0)=[0.767,0.685]^\transpose$ [m] with an initial heading of $1.553$ [rad] and $-1.574$ [rad], respectively. The leader is assigned a constant linear speed $v_1 (t) \equiv 0.08$ [m/s] with a time-varying angular speed  $u_1(t)=0.2+0.5\sin(0.5\pi t)$ [rad/s]. 
The follower employs the leader-agnostic feedback control (Prop.~\ref{prop:2}), with desired $\rho_0 = 0.2$ [m] and gains $\mu_1 = 1.5, \mu_2 = 2$, and $\mu_\rho = 20$.

The experimental results are provided in Fig.~\ref{fig:experiment}. 
Even with the lack of information about the leader's steering $u_1$, the follower still successfully tracks the leader, converging to a periodic orbit near the desired abreast formation. All code used for numerical simulation and robotic implementation, including a full video demonstration of all experiments, can be found at \cite{Github2026Speed}.


\begin{figure*}[t]
	\centering
	\includegraphics[width=1\textwidth]{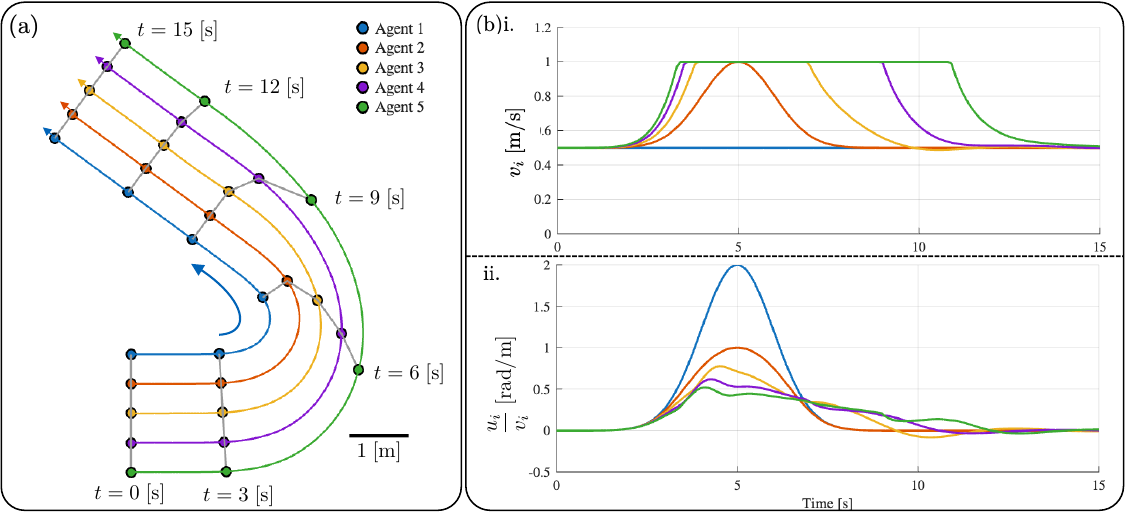}
	\caption{Formation control of $N$ agents ($N = 5$), connected in a chain network as described in Sec.~\ref{sec:nagents}. 
    The trajectories of the $N$ agents at different times are shown in (a), with a line in gray connecting each agent to highlight the whip-like motion taken by the formation of the agents. The speed and steering controls of each agent are plotted in (b)(i)-(ii) (with the same color coding as in (a)).}
	\label{fig:Nagents}
\end{figure*}

\section{Extension to $N$-agent Case} \label{sec:nagents}
The two-agent framework can be naturally extended to a multi-agent system. In this section, a simple extension to an $N$-agent system is provided. Here, agent 1 is considered the `global leader', and each agent $i+1$ treats the preceding agent $i$ as its leader, $i = 1, ..., N-1$, forming a chain network. 

 The dynamics of the network ($N = 5$) are numerically simulated in MATLAB using the same numerical integrator as in Sec.~\ref{sec:numerics} for each $(i, i+1)$ agent pair ($i=1, ..., 4$). The agents are initialized at positions $\mathbf{r}_i(0) = [0, -0.5i]^\transpose$ [m] with uniform headings of 0 [rad]. 
 Agent 1, the global leader, moves at a constant speed $v_1 (t) \equiv 0.5$ [m/s]. To evaluate the transient response of the chain, Agent~1 operates with an initial turning rate of $u_1(0) = 0$ [rad/s] but receives a Gaussian steering impulse with amplitude 1 [rad/s] and standard deviation 1 [s], centered at $5$ [s]. Each of the remaining agents follows the speed and steering control laws as in Prop. \ref{prop:1}, and is limited to a maximum speed of $v_{i, \text{max}} = 1$ [m/s]. 

The simulation results are depicted in Fig.~\ref{fig:Nagents}. The steering impulse propagates down the chain, with each subsequent follower adapting its trajectory to maintain the Bertrand mate formation relative to its immediate predecessor. As agents on the `outside' need to move faster to keep in formation, they reach the maximum speed, resulting in a \textit{whip or wave-like motion} of the formation of the agents (see Fig.~\ref{fig:Nagents}(a)). The curvatures of the agents' trajectories, shown in Fig. \ref{fig:Nagents}(b), illustrate how the perturbation is propagated through the network, with each agent returning to equilibrium after some delay. This emergent behavior of the multi-agent system is prominently seen in natural flocking as well~\cite{attanasi2014information}, with recent research modeling the propagation of these wave-like structures in biological swarms \cite{OKeeffe2017, halder2019optimality}.
\vspace{-4pt}
\section{Conclusion} \label{sec:conclusion}
\vspace{-2pt}
In this paper, we highlight the role of speed regulation in maintaining formation within a two-agent system. Using the unicycle model for the agents, we develop a cascade feedback control strategy that couples a novel speed control law with a constant bearing steering law. Stability analysis demonstrates that this approach leads to asymptotic convergence of the follower to a targeted abreast formation when the leader's steering is known. When the leader's steering inputs are unknown, we show that the speed control is still robust by proving input-to-state stability. Further, we prove that even periodic inputs still lead to convergence to a periodic orbit in the follower's trajectory. These theoretical guarantees are successfully validated by both numerical simulations and hardware experiments using mobile robots. Additionally, by extending the leader-follower paradigm to an $N$-agent chain network, we provide a mechanism for modeling movement in a larger collective. 

Future work will explore the relaxation of the orthogonal bearing constraint to accommodate arbitrary desired relative headings (i.e., other than $\pm \frac{\pi}{2}$), enabling a generalized full constant bearing control framework. We will also analyze possible control algorithms which allow online switching between the left and right following convention for the follower while preserving stability. Additionally, we will consider multi-agent formations beyond simple chain networks. Specifically, we plan to investigate distributed control across complex network topologies where an agent determines its trajectory by following a virtual leader synthesized from the spatial average of its $k$-nearest neighbors.

\vspace*{-0pt}
\bibliographystyle{IEEEtran}
\bibliography{bibfiles/reference, bibfiles/octopus_papers, bibfiles/collectives_motivation}

@article{couzin2002collective,
  title={Collective memory and spatial sorting in animal groups},
  author={Couzin, Iain D and Krause, Jens},
  journal={Journal of Theoretical Biology},
  volume={218},
  number={1},
  pages={1--11},
  year={2002}
}

@article{katz2011inferring,
  title={Inferring the structure and dynamics of interactions in schooling fish},
  author={Katz, Yael and Tunstr{\o}m, Kolbj{\o}rn and Ioannou, Christos C and Huepe, Cristi{\'a}n and Couzin, Iain D},
  journal={Proceedings of the National Academy of Sciences},
  volume={108},
  number={46},
  pages={18720--18725},
  year={2011},
  publisher={National Academy of Sciences}
}

@article{olfati2006flocking,
  title={Flocking for multi-agent dynamic systems: Algorithms and theory},
  author={Olfati-Saber, Reza},
  journal={IEEE Transactions on Automatic Control},
  volume={51},
  number={3},
  pages={401--420},
  year={2006},
  publisher={IEEE}
}

@article{paley2007oscillator,
  title={Oscillator models and collective motion},
  author={Paley, Derek A and Leonard, Naomi Ehrich and Sepulchre, Rodolphe and Grunbaum, Daniel and Parrish, Julia K},
  journal={IEEE Control Systems Magazine},
  volume={27},
  number={4},
  pages={89--105},
  year={2007},
  publisher={IEEE}
}

@article{sepulchre2007stabilization,
  title={Stabilization of planar collective motion: All-to-all communication},
  author={Sepulchre, Rodolphe and Paley, Derek A and Leonard, Naomi Ehrich},
  journal={IEEE Transactions on Automatic Control},
  volume={52},
  number={5},
  pages={811--824},
  year={2007},
  publisher={IEEE}
}

@article{zago2018speed,
  title={The speed-curvature power law of movements: a reappraisal},
  author={Zago, Myrka and Matic, Adam and Flash, Tamar and Gomez-Marin, Alex and Lacquaniti, Francesco},
  journal={Experimental Brain Research},
  volume={236},
  number={1},
  pages={69--82},
  year={2018},
  publisher={Springer}
}

@book{lavalle2006planning,
  title={Planning algorithms},
  author={LaValle, Steven M},
  year={2006},
  publisher={Cambridge University Press}
}

@IEEEtranBSTCTL{BSTcontrol,
CTLuse_forced_etal       = "yes",
CTLmax_names_forced_etal = "2",
CTLnames_show_etal       = "1" }

@book{khalil2002nonlinear,
  title={Nonlinear Systems},
  author={Khalil, Hassan K},
  publisher={Prentice Hall},
  year={2002}
}

@inproceedings{halder2016steering,
	title={Steering for beacon pursuit under limited sensing},
	author={Halder, Udit and Schlotfeldt, Brent and Krishnaprasad, PS},
	booktitle={Proceedings of the 55th IEEE Conference on Decision and Control (CDC)},
	pages={3848--3855},
	year={2016},
	organization={IEEE}
}

@article{galloway2013symmetry,
  title={Symmetry and reduction in collectives: cyclic pursuit strategies},
  author={Galloway, Kevin S and Justh, Eric W and Krishnaprasad, PS},
  journal={Proceedings of the Royal Society A: Mathematical, Physical and Engineering Sciences},
  volume={469},
  number={2158},
  pages={20130264},
  year={2013},
  publisher={The Royal Society Publishing}
}

@article{justh2006steering,
	title={Steering laws for motion camouflage},
	author={Justh, Eric W and Krishnaprasad, PS},
	journal={Proceedings of the Royal Society A: Mathematical, Physical and Engineering Sciences},
	volume={462},
	number={2076},
	pages={3629--3643},
	year={2006},
	publisher={The Royal Society London}
}

@phdthesis{halder2019optimality,
  title={Optimality, Synthesis and a Continuum Model for Collective Motion},
  author={Halder, Udit},
  school={University of {M}aryland, {C}ollege {P}ark},
  year={2019}
}

@article{flash1985coordination,
  title={The coordination of arm movements: an experimentally confirmed mathematical model},
  author={Flash, Tamar and Hogan, Neville},
  journal={Journal of Neuroscience},
  volume={5},
  number={7},
  pages={1688--1703},
  year={1985},
  publisher={Soc Neuroscience}
}

@article{attanasi2014information,
  title={Information transfer and behavioural inertia in starling flocks},
  author={Attanasi, Alessandro and Cavagna, Andrea and Del Castello, Lorenzo and Giardina, Irene and Grigera, Tomas S and Jeli{\'c}, Asja and Melillo, Stefania and Parisi, Leonardo and Pohl, Oliver and Shen, Edward and others},
  journal={{Nature Physics}},
  volume={10},
  number={9},
  pages={691--696},
  year={2014},
  publisher={Nature Publishing Group}
}

@article{mischiati2015internal,
  title={Internal models direct dragonfly interception steering},
  author={Mischiati, Matteo and Lin, Huai-Ti and Herold, Paul and Imler, Elliot and Olberg, Robert and Leonardo, Anthony},
  journal={Nature},
  volume={517},
  number={7534},
  pages={333--338},
  year={2015},
  publisher={Nature Publishing Group}
}

@article{Bishop1975,
  author  = {Bishop, Richard L.},
  title   = {There is more than one way to frame a curve},
  journal = {The American Mathematical Monthly},
  volume  = {82},
  number  = {3},
  pages   = {246--251},
  year    = {1975},
  publisher = {Taylor \& Francis}
}

@inproceedings{Zhang2004,
  author    = {Zhang, Fumin and Justh, Eric and Krishnaprasad, P. S.},
  title     = {Boundary following using gyroscopic control},
  booktitle = {Proceedings of the 43rd IEEE Conference on Decision and Control (CDC)},
  volume    = {5},
  pages     = {5204--5209},
  year      = {2004},
  month     = {Dec},
  doi       = {10.1109/CDC.2004.1429634}
}

@article{Murphy2019,
  author  = {Murphy, David W. and Olsen, Daniel and Kanagawa, Marleen and King, Rob and Kawaguchi, So and Osborn, Jon and Webster, Donald R. and Yen, Jeannette},
  title   = {The Three Dimensional Spatial Structure of Antarctic Krill Schools in the Laboratory},
  journal = {Scientific Reports},
  year    = {2019},
  volume  = {9},
  doi     = {10.1038/s41598-018-37379-9}
}

@article{brighton2017,
  title={Terminal attack trajectories of peregrine falcons are described by the proportional navigation guidance law of missiles},
  author={Brighton, Caroline H. and Thomas, Adrian L. R. and Taylor, Graham K.},
  journal={Proceedings of the National Academy of Sciences},
  volume={114},
  number={51},
  year={2017},
  pages={13495--13500}
}

@manual{turtlebot3_manual, 
  author = {{ROBOTIS}},
  title = {{TurtleBot3}},
  organization = {ROBOTIS Co., Ltd.},
  year = {2026},
  note = {Available online: \url{https://emanual.robotis.com/docs/en/platform/turtlebot3/overview/}} }

@misc{optitrack_website,
  author       = {{NaturalPoint, Inc.}},
  title        = {{OptiTrack} Motion Capture Systems and {Motive} Software},
  howpublished = {\url{https://optitrack.com/}},
  year         = {2026},
}

@book{Davis1962,
  author    = {Davis, Harold T.},
  title     = {Introduction to Nonlinear Differential and Integral Equations},
  publisher = {Dover Publications},
  year      = {1962},
  address   = {New York}
}

@article{Bertrand1850,
  author  = {Bertrand, J.},
  title   = {Mémoire sur la théorie des courbes à double courbure},
  journal = {Journal de Math\'{e}matiques Pures et Appliqu\'{e}es},
  series  = {1},
  volume  = {15},
  pages   = {332--350},
  year    = {1850}
}

@article{OKeeffe2017,
  author  = {O'Keeffe, Kevin P. and Hong, Hyunsuk and Strogatz, Steven H.},
  title   = {Oscillators that sync and swarm},
  journal = {Nature Communications},
  volume  = {8},
  number  = {1},
  pages   = {1504},
  year    = {2017},
  doi     = {10.1038/s41467-017-01190-3}
}

@article{tang2021formation,
  title={Formation control of a leader--follower structure in three dimensional space using bearing measurements},
  author={ Tang, Z. and Cunha, R. and Hamel, T. and Silvestre, C.},
  journal={Automatica},
  volume={128},
  pages={109567},
  year={2021},
  publisher={Elsevier}
}

@article{trinh2021finite,
  title={Finite-Time Bearing-Based Maneuver of Acyclic Leader-Follower Formations},
  author={Trinh, M. H. and Ahn, H.-S.},
  journal={IEEE Control Systems Letters},
  volume={6},
  pages={1004--1009},
  year={2021},
  publisher={IEEE}
}

@article{consolini2008leader,
  title={Leader--follower formation control of nonholonomic mobile robots with input constraints},
  author={Consolini, L. and Morbidi, F. and Prattichizzo, D. and Tosques, M.},
  journal={Automatica},
  volume={44},
  number={5},
  pages={1343--1349},
  year={2008},
  publisher={Elsevier}
}

@article{do2007formation,
  title={Formation tracking control of unicycle-type mobile robots with limited sensing ranges},
  author={{Do, K.-D.}},
  journal={IEEE Transactions on Control Systems Technology},
  volume={15},
  number={3},
  pages={526--538},
  year={2007},
  publisher={IEEE}
}

@misc{Github2026Speed,
  author       = {Li, Xincheng and Liu, Tengyue and Halder, Udit},
  title        = {GitHub repository},
  year         = {2026},
  publisher    = {GitHub},
  howpublished = {\url{https://github.com/HalderLab/cdc2026-Speed-Control}},
}

\appendices
\renewcommand{\thelemma}{A-\arabic{section}.\arabic{lemma}}
\renewcommand{\thetheorem}{A-\arabic{section}.\arabic{theorem}}
\renewcommand{\theequation}{A-\arabic{equation}}
\renewcommand{\thedefinition}{A-\arabic{definition}}
\setcounter{lemma}{0}
\setcounter{theorem}{0}
\setcounter{equation}{0}

\section{Leader Parameter Estimation and Stability Analysis} \label{appdx:estimation_analysis}
Let $\hat{\dot{\rho}}$ and $\hat{\dot{\alpha}}_2$ denote the estimates for $\dot\rho$ and $\dot\alpha_2$, respectively, which we assume possess bounded errors from sensor or processing noise. In other words, there exists $M_\rho > 0$, $ M_{\alpha_2} > 0$ such that
\begin{equation}
    |\dot{\rho} - \hat{\dot{\rho}}| \le M_\rho, \quad |\dot{\alpha}_2 - \hat{\dot{\alpha}}_2| \le M_{\alpha_2}
\end{equation}

Using our rate estimates, we construct the corresponding Cartesian estimates $\hat{X}$ and $\hat{Y}$
\begin{equation} 
    \hat{X} = - \hat{\dot\rho} - v_2 \cos\alpha_2, \quad \hat{Y} = \rho(\hat{\dot\alpha}_2 + u_2) - v_2\sin\alpha_2 
\end{equation}
Here $v_2$ and $u_2$ are the measured physical states of the follower, not the control laws which we design. The estimation errors for these terms are then strictly bounded by
\begin{equation}
    |X - \hat{X}| \le M_\rho, \quad |Y - \hat{Y}| \le \rho M_{\alpha_2}
\end{equation}

We now derive estimates for the leader's speed $\hat{v}_1$ and bearing $\hat{\alpha}_1$ from the Cartesian estimates:
\begin{equation}
    \hat{v}_1 = \sqrt{\hat{X}^2 + \hat{Y}^2}, \quad \hat{\alpha}_1 = \arctan\left(\frac{\hat{Y}}{\hat{X}}\right)
\end{equation}
Noting that $v_1 = \sqrt{X^2 + Y^2}$ and $\alpha_1 = \arctan\left(\frac{Y}{X}\right)$ can be thought of as polar coordinates of the Cartesian terms, we can apply a first-order Taylor series expansion to map the Cartesian errors to these polar coordinates. After some standard calculations, the estimation errors $\tilde{v}_1 = v_1 - \hat{v}_1$ and $\tilde{\alpha}_1 = \alpha_1 - \hat{\alpha}_1$ can be shown to be bounded by:
\begin{align}
\begin{split}
    |\tilde{v}_1| &\le \sqrt{M_\rho^2 + \rho^2 M_{\alpha_2}^2} =: M_{v_1}(\rho) \\
|\tilde{\alpha}_1| &\le \frac{1}{v_\epsilon} \sqrt{M_\rho^2 + \rho^2 M_{\alpha_2}^2} = \frac{1}{v_\epsilon} M_{v_1}(\rho)
\end{split}
\end{align}

With these estimates, the dynamic of the system are now given by 
\begin{align} \label{eq: estimated_dynamics}
    \begin{split}
        \dot{\rho} &= -v_1 \cos \alpha_1 - \hat{v}_2 \cos \alpha_2 \\
        \dot{\alpha}_1 &= - u_1 + \frac{1}{\rho} (v_1 \sin \alpha_1 + \hat{v}_2 \sin \alpha_2)  \\
        \dot{\alpha}_2 &= - \hat{u}_2 + \frac{1}{\rho} (v_1 \sin \alpha_1 + \hat{v}_2 \sin \alpha_2) 
    \end{split}
\end{align}
where $\hat{v}_2$ and $\hat{u}_2$ are the feedback controls constructed using the $\hat{v}_1$ and $\hat{\alpha}_1$ estimations.

Using these estimates, we construct the steering control law $\hat{u}_2$. 
\begin{equation}
    \hat{u}_2 = -\mu_2 \cos\alpha_2 + \frac{1}{\rho}(\hat{v}_1\sin\hat\alpha_1 + \hat{v}_2 \sin\alpha_2)
\end{equation}
When $\hat{u}_2$ is substituted into the $\alpha_2$ dynamics \eqref{eq: estimated_dynamics}, the resulting closed-loop system takes the form of a nominal system subject to a perturbation:
\begin{equation}
    \dot{\alpha}_2 = \mu_2 \cos \alpha_2 + (\dot{\alpha}_2 - \hat{\dot{\alpha}}_2)
\end{equation}
Because the nominal system $\dot{\alpha}_2 = \mu_2 \cos \alpha_2$ is exponentially stable, it is Input-to-State Stable (ISS) with respect to the error $\dot{\alpha}_2 - \hat{\dot{\alpha}}_2$, which is bounded by $M_{\alpha_2}$. The perturbation is strictly bounded by $M_{\alpha_2}$. Consequently, $\alpha_2$ will still exponentially converge to an $\epsilon$-ball around the desired equilibrium $\alpha_2 = \pi/2$, where the ultimate bound $\epsilon$ is proportional to $M_{\alpha_2} / \mu_2$.

To analyze the stability of the subsystem $z_1 = (\rho, \alpha_1)$, we construct our speed control law $\hat{v}_2$ using these estimated states ($\hat{v}_1$ and $\hat{\alpha}_1$):
\begin{equation}
    \hat{v}_2 = -\hat{v}_1 \sin \hat{\alpha}_1 + \rho \left( -\mu_1 \cos\hat{\alpha}_1 + \hat{v}_1 f(\rho) \right) 
\end{equation}
The resulting closed-loop dynamics for the $z_1$ subsystem can be written as:
\begin{equation}
    \dot{z}_1 = h(z_1, v_2) + \tilde{h}
\end{equation}
where $h(z_1, v_2)$ is the nominal, unperturbed system (same as $h_1$ in the proof of Prop. 3.3), and $\tilde{h} = \hat{h} - h$ represents the system error injected by the estimations.

The magnitude of the system error $\tilde{h}$ is bounded by:
\begin{equation}
    \|\tilde{h}\| \le  M_{\alpha_2} + \mu_1 + f(\rho)M_{v_1}(\rho)
\end{equation}
While this bound explicitly depends on $\rho$, we may assume the system operates within a compact set $\rho \in [\rho_{min}, \rho_{max}]$, that is, our initial conditions are such that we operate within a compact, forward-invariant region of attraction around the equilibrium. Because $f(\rho)$ is continuous, the term $f(\rho)M_{v_1}(\rho)$ reaches a finite maximum on this interval. Therefore, the entire perturbation vector is globally upper-bounded by some constant $M_h$ within this set.

Thus, the $z_1$ subsystem is a nominal system subject to a strictly bounded perturbation on the $\alpha_1$ dynamics. As proven in Proposition 3.2, the nominal system is ISS. Therefore we can follow the same arguments to show that the entire system guarantees bounded stability around the equilibrium, completing the stability analysis for the estimation of $v_1$ and $\alpha_1$.

\end{document}